\documentclass[conference,letterpaper]{IEEEtran}

\usepackage[T1]{fontenc}% optional T1 font encoding
\usepackage{amssymb}
\usepackage[utf8]{inputenc}
\usepackage[english]{babel}
\usepackage{amsmath}			
\usepackage{courier}

\usepackage[left=0.6in, right=0.6in, top=0.65in, bottom=0.6in]{geometry}

\usepackage{comment}				
\usepackage[T1]{fontenc}
\usepackage{mathtools}
\usepackage{algcompatible,amsmath}
\usepackage[ruled,vlined]{algorithm2e}
\usepackage{graphicx}
\usepackage{multicol}
\usepackage{afterpage}
\usepackage[noend]{algpseudocode}
\usepackage{color}
\usepackage{etoolbox}
\usepackage{url}
\usepackage{siunitx,booktabs}
\usepackage{placeins}
\usepackage{bbm}
\usepackage{soul}
\usepackage[noadjust]{cite}
\DeclareMathOperator{\tr}{tr}
\usepackage{amsfonts}

\usepackage{accents}

\normalsize
\makeatletter
\def\underbracex#1#2{\mathop{\vtop{\m@th\ialign{##\crcr
				$\hfil\displaystyle{#2}\hfil$\crcr
				\noalign{\kern3\p@\nointerlineskip}%
				#1\crcr\noalign{\kern3\p@}}}}\limits}

\def\upbracefilla{$\m@th \setbox\z@\hbox{$\braceld$}%
	\bracelu\leaders\vrule \@height\ht\z@ \@depth\z@\hfill 
	\kern\p@\vrule \@width\p@\kern\p@\vrule \@width\p@\kern\p@\vrule \@width\p@
	$}

\def\upbracefillb{$\m@th \setbox\z@\hbox{$\braceld$}%
	\vrule \@width\p@\kern\p@\vrule \@width\p@\kern\p@\vrule \@width\p@\kern\p@
	\leaders\vrule \@height\ht\z@ \@depth\z@\hfill\bracerd
	\braceld\leaders\vrule \@height\ht\z@ \@depth\z@\hfill
	\kern\p@\vrule \@width\p@\kern\p@\vrule \@width\p@\kern\p@\vrule \@width\p@
	$}

\def\upbracefillc{$\m@th \setbox\z@\hbox{$\braceld$}%
	\vrule \@width\p@\kern\p@\vrule \@width\p@\kern\p@\vrule \@width\p@\kern\p@
	\leaders\vrule \@height\ht\z@ \@depth\z@\hfill
	\kern\p@\vrule \@width\p@\kern\p@\vrule \@width\p@\kern\p@\vrule \@width\p@
	$}

\def\upbracefilld{$\m@th \setbox\z@\hbox{$\braceld$}%
	\vrule \@width\p@\kern\p@\vrule \@width\p@\kern\p@\vrule \@width\p@\kern\p@
	\leaders\vrule \@height\ht\z@ \@depth\z@\hfill\braceru$}

\def\underbracex#1#2{\mathop{\vtop{\m@th\ialign{##\crcr
				$\hfil\displaystyle{#2}\hfil$\crcr
				\noalign{\kern3\p@\nointerlineskip}%
				#1\crcr\noalign{\kern3\p@}}}}\limits}

\def\upbracefilla{$\m@th \setbox\z@\hbox{$\braceld$}%
	\bracelu\leaders\vrule \@height\ht\z@ \@depth\z@\hfill 
	\kern\p@\vrule \@width\p@\kern\p@\vrule \@width\p@\kern\p@\vrule \@width\p@
	$}

\def\upbracefillb{$\m@th \setbox\z@\hbox{$\braceld$}%
	\vrule \@width\p@\kern\p@\vrule \@width\p@\kern\p@\vrule \@width\p@\kern\p@
	\leaders\vrule \@height\ht\z@ \@depth\z@\hfill\bracerd
	\braceld\leaders\vrule \@height\ht\z@ \@depth\z@\hfill
	\kern\p@\vrule \@width\p@\kern\p@\vrule \@width\p@\kern\p@\vrule \@width\p@
	$}

\def\upbracefillbd{$\m@th \setbox\z@\hbox{$\braceld$}%
	\vrule \@width\p@\kern\p@\vrule \@width\p@\kern\p@\vrule \@width\p@\kern\p@
	\bracerd\braceld
	\leaders\vrule \@height\ht\z@ \@depth\z@\hfill\braceru$}

\makeatletter
\patchcmd{\@makecaption}
{\scshape}
{}
{} 
{}
\makeatother
\usepackage{subcaption}
\usepackage{graphicx}
\usepackage{amsthm}

\newtheorem{proposition}{Proposition}

\newtheorem*{proof*}{Proof}

\newcommand{\norm}[1]{\Big\lVert#1\Big\rVert}

\IEEEoverridecommandlockouts
\makeatletter
\let\origIEEEPARstart\IEEEPARstart
\renewcommand{\IEEEPARstart}[3][1.1]{%
	\def\@IEEEPARstartDROPDEPTH{#1\baselineskip}%
	\origIEEEPARstart{#2}{#3}%
}

\def\endthebibliography{%
	\def\@noitemerr{\@latex@warning{Empty `thebibliography' environment}}%
	\endlist
}
\makeatletter
\newcommand{\reducespacing}{\def\@IEEEaftertitletext{}}
\makeatother

\ifCLASSINFOpdf

\else

\fi

\begin{document}
	\bstctlcite{BSTcontrol}
	\setlength{\columnsep}{0.21in}   % place this right here, before \maketitle

	\title{Mask-Compliant Clipping-Aware Precoding for Multi-User MIMO-OFDM Systems}

	\author{
		\IEEEauthorblockN{Navid Reyhanian\IEEEauthorrefmark{1},
			Parisa Ramezani\IEEEauthorrefmark{2},
			Emil Bj\"ornson\IEEEauthorrefmark{2}}
		\IEEEauthorblockA{\IEEEauthorrefmark{1}Cisco Systems, Milpitas, CA 95035, USA}
		\IEEEauthorblockA{\IEEEauthorrefmark{2}Department of Communication Systems, KTH Royal Institute of Technology, Stockholm, Sweden}
		\IEEEauthorblockA{nreyhani@cisco.com, \{parram,emilbjo\}@kth.se}
	}	

	\allowdisplaybreaks

	\maketitle
	\setlength{\columnsep}{0.21in}   % MUST be after \maketitle
	\thispagestyle{empty}
	\pagestyle{empty}
	
	\begin{abstract}
		Orthogonal frequency-division multiplexing (OFDM) is widely adopted in frequency-selective channels for its ability to simplify equalization, yet it also causes large signal peaks, out-of-band (OOB) emissions, and spectral leakage. We study downlink precoding/combining for multi-user multiple-input multiple-output (MU-MIMO)--OFDM systems by minimizing the sum of the users' mean-squared errors (MSEs) under per-subcarrier transmit-power limits, per-antenna OOB spectral-mask constraints, and per-antenna peak-amplitude (clipping) constraints, which confine the emitted spectrum and limit waveform peaks to reduce power-amplifier saturation, nonlinear distortion, and spectral regrowth. The resulting nonconvex problem is handled by a minimum mean-squared error (MMSE) based block coordinate descent (BCD), with a closed-form combiner update and an alternating direction method of multipliers (ADMM) algorithm with closed-form updates for the precoder subproblem. Simulations show clear gains over well-known benchmark schemes.
	\end{abstract}
	
	\IEEEpeerreviewmaketitle

	\section{Introduction}
	\label{sec:introduction}
Frequency-selective channels pose significant challenges for reliable wireless communication. Orthogonal frequency-division multiplexing (OFDM) addresses this by converting the channel into parallel subcarriers, thereby simplifying equalization. However, OFDM still faces practical transmitter impairments \cite{reyhanian2025precodinguplinkrisassistedcellfree}. For example, its time-domain signal may exhibit large instantaneous peaks, and the resulting multicarrier waveform can produce sidelobes that cause out-of-band (OOB) emissions and spectral leakage \cite{7366560}. These effects are especially important in practice, since OOB radiation may interfere with adjacent channels and can violate regulatory or standard-imposed spectral limits. To further improve spectral efficiency and manage interference, multiple-input multiple-output (MIMO) deployments with large antenna arrays are increasingly adopted. In particular, multi-user MIMO (MU-MIMO) enables a base station (BS) to serve multiple users simultaneously on the same time–frequency resources, but doing so reliably requires precoding and receiver combining that both provide array gain and actively control inter-user interference \cite{11588254}.
	
	To keep OFDM signals within the operating range of digital-to-analog converters and, in particular, power amplifiers (PAs), amplitude clipping, or equivalently per-antenna peak-amplitude constraints, is often used. Although this helps prevent PA saturation and severe nonlinear distortion, it can also cause in-band distortion and spectral regrowth that worsens OOB leakage \cite{9151131}. Practical designs must therefore jointly manage multi-user interference, waveform peaks, and per-antenna spectral masks such as adjacent channel leakage ratio (ACLR) or mask-based limits \cite{9390405,10048700,reyhanian2026jointspatialspectralhybrid,reyhanian2026symbollevelmaskcomplianthybridprecoding}.
	
A substantial literature studies spectral shaping for realized OFDM symbols, including mask-compliant precoding and notching for suppressing OOB emissions at selected frequencies \cite{van2009sculpting,6459499,7485853,9214877}. These methods have been combined with MIMO precoding schemes such as zero-forcing (ZF) and maximum-ratio transmission (MRT) \cite{taheri2020joint,9214877,9785463}. 
However, existing approaches separate spatial and spectral processing and do not jointly enforce per-antenna mask and peak constraints. In MU-MIMO-OFDM, cascading independently-designed spectral and spatial precoders distorts the spatially-precoded signals, amplifying inter-user interference and degrading spectral efficiency.

Existing regulations constrain instantaneous peak emissions, both in-band and out-of-band. In particular, the ultra wideband (UWB) rule specifies a peak-power ceiling of 0\,dBm/50\,MHz~\cite{fcc_uwb_15_517}, and Section~96.41 for citizens broadband radio service (CBRS) requires OOB compliance to be assessed through peak-detected measurements~\cite{fcc_cbrs_96_41}. This is especially relevant here because nonlinear analog components placed after the digital precoder can generate spectral regrowth. Therefore, to ensure regulatory compliance in the transmitted waveform, the OOB spectral mask should be imposed on the emitted spectrum of each OFDM symbol \cite{9214877}.
	
	In this paper, we develop a unified downlink MU-MIMO-OFDM precoding/combining framework that employs symbol-specific precoding and symbol-agnostic receive combining to minimize the user sum mean-squared error (sum-MSE) under explicit per-subcarrier transmit-power constraints, per-antenna peak-amplitude (clipping) limits for waveform peak control and PA saturation mitigation, and per-antenna OOB spectral-mask constraints for each realized symbol vector. In contrast to many existing approaches that mainly aim to reduce OOB emissions as much as possible, our framework is designed to guarantee compliance with arbitrary required spectral limits.
	The resulting joint design tightly couples the spatial and spectral dimensions and
	yields a challenging nonconvex optimization.
	To solve this problem efficiently, we propose a minimum mean squared error (MMSE)-driven block coordinate descent (BCD) algorithm that alternates between updating combiners and precoders. A closed-form solution is derived for the combiner. The precoder step is solved using a scalable alternating direction method of
	multipliers (ADMM) that splits the update into several simpler substeps by introducing auxiliary variables: one set captures the
	frequency-domain precoded signals, another represents the corresponding time-domain waveforms used to enforce peak (clipping) limits,
	and a third represents the sampled spectrum used to enforce the OOB emission mask. This splitting turns a large constrained problem
	into solving a sequence of easy subproblems with closed-form solutions or one-dimensional bisection-searches. We further provide convergence guarantees
	and demonstrate via simulations that the proposed method achieves strict OOB mask compliance while delivering clear
	performance gains over benchmark MU-MIMO-OFDM precoding schemes.

\section{System Model}
We study the downlink of an MU-MIMO-OFDM system in which a BS with $N_t$ transmit antennas serves $K$ users over $S$ subcarriers.
The subcarrier index set is defined as $\mathcal{S}=\{0,1,\ldots,S-1\}$. Each user is equipped with $N_r$ receive antennas and is
scheduled on all subcarriers. A fully digital architecture is assumed at the transmitter and receiver. For subcarrier $s$, the precoder used to deliver $n_k$ streams to user $k$ is
$\mathbf{V}_{k}^s\in\mathbb{C}^{N_t \times n_k }$, and the corresponding transmitted symbol vector is
$\boldsymbol{\omega}_{k}^s \in \mathbb{C}^{n_k \times 1}$. We treat $\boldsymbol{\omega}_{k}^s\in \boldsymbol{\Omega}_0$ as a zero-mean normalized QAM symbol vector with
$\mathbb E[\boldsymbol{\omega}_j^s(\boldsymbol{\omega}_j^s)^H]=\mathbf I_{n_j}$ and
$\mathbb E[\boldsymbol{\omega}_j^s(\boldsymbol{\omega}_\ell^s)^H]=\mathbf 0$ for $j\neq \ell$. 
Once realized, $\boldsymbol{\omega}_{k}^s$ is assumed known at the transmitter, which allows symbol-dependent precoding.
We consider a batch of $B$ independent symbol realizations $\{\boldsymbol{\omega}^{(b)}\}_{b=1}^B$, $\boldsymbol{\omega}^{(b)}\triangleq\{\boldsymbol\omega_k^{s,(b)}\}_{k,s}\in \boldsymbol{\Omega}$, and for each $b$ the BS designs $\mathbf V_k^{s,(b)}$ and forms $
\mathbf t^{s,(b)} \triangleq \sum_{j=1}^K \mathbf V_j^{s,(b)}\boldsymbol\omega_j^{s,(b)} \in \mathbb C^{N_t}.$ The combiners $\{\mathbf U_k^s\}$ are shared across $b$. Below, the superscript $(b)$ is shown only where needed.

The received signal
$\mathbf{y}_{k}^{s,(b)} \in \mathbb{C}^{N_r \times 1}$ at user $k$ on subcarrier $s$ is expressed as
	\begin{equation}
		\mathbf{y}_k^{s,(b)} 
		= \mathbf{H}_k^s \sum_{j=1}^K \mathbf{V}_j^{s,(b)} \boldsymbol{\omega}_j^{s,(b)} + \mathbf{n}_k^{s,(b)}=\mathbf H_k^s\mathbf t^{s,(b)}+\mathbf n_k^{s,(b)},
		\nonumber
	\end{equation}
where $\mathbf{H}_k^s \in \mathbb{C}^{N_r \times N_t}$ denotes the channel from the BS to user $k$ on subcarrier $s$, and
$\mathbf{n}_k^{s,(b)} \in \mathbb{C}^{N_r \times 1}$ is additive white  Gaussian noise distributed as
$\mathcal{CN}(\mathbf{0}, \sigma_{\text{noise},k}^{s\,2}\mathbf{I}_{N_r})$.
At the receiver, user $k$ employs a combiner. The decoded signal vector on subcarrier $s$ at the $k^\text{th}$ user is $\hat{\boldsymbol{\omega}}_k^{s,(b)}
= \mathbf{U}_k^{s^H}\mathbf{y}_k^{s,(b)},
\mathbf{U}_{k}^s\in\mathbb{C}^{N_r \times n_k }.$

\subsubsection{Power Constraint for Precoding}
The BS transmit budget on each subcarrier is limited according to
	\begin{align}
		\norm{\sum_{j=1}^K \mathbf{V}_j^{s,(b)} \boldsymbol{\omega}_j^{s,(b)}}_2^2=\|\mathbf t^{s,(b)}\|_2^2
		\leq P^s,
		\label{eq:powerbudget_digital}
	\end{align}
where $P^s$ is the total BS power on subcarrier $s$. 

\subsubsection{Clipping Constraint}

With an oversampling factor $\ell$, the cyclic-prefix (CP) inclusive OFDM symbol emitted by the $a^\text{th}$ antenna is described in discrete time over
$n\in\{-\ell N_{\mathrm{CP}},-\ell N_{\mathrm{CP}}+1,\ldots,\ell S-1\}$. Over this interval, the transmitted samples follow \cite{van2009sculpting}
\begin{equation}
	x_{\text{CP}}^{a,(b)}[n] = \mathbf{p}^T[n]\mathbf{g}^{a,(b)}
	= \sum_{s\in \mathcal{S}} p^{s}[n] \mathbf{g}^{a,(b)}[s],
	\label{eq:ofdm_sym_digital}
\end{equation}
where $\mathbf{p}[n] = [p^{0}[n], \ldots, p^{S-1}[n]]^T$ collects the subcarrier-modulated pulses (sampled on the $\ell S$-point grid) and
$\mathbf{g}^{a,(b)}\in\mathbb{C}^{S}$ stacks the corresponding precoded frequency-domain symbols associated with the $a^\text{th}$ antenna.
We have
\begin{equation}
	\mathbf g^{a,(b)}[s]=\mathbf t^{s,(b)}[a],\qquad a=1,\ldots,N_t,\ \ s\in\mathcal S .
	\label{eq:ga_def_digital}
\end{equation}

The discrete-time OFDM pulse with subcarrier modulation for the $s^\text{th}$ subcarrier is defined as \cite{van2009sculpting,9214877,9390405}
\begin{equation}
	p^s[n] = \frac{1}{\sqrt{\ell S}}\,e^{\jmath 2\pi \frac{s}{\ell S} n}\, I[n], \:\:\: s\in \mathcal{S},
	\label{eq:timedomain_digital}
\end{equation}
where $I[n]$ is the indicator function given by $I[n] = 1$ for $-\ell N_{\mathrm{CP}} \leq n \leq \ell S-1$ and $I[n] = 0$ otherwise,
with $N_{\mathrm{CP}}$ representing the CP length (in samples). The rectangular window assumption provides a worst-case
characterization of OOB emissions since it yields the slowest spectral decay among common pulse shapes; thus, any suppression
demonstrated here is expected to be at least as good when band-limited pulses are employed.

Fix $a\in\{1,\ldots,N_t\}$ and consider an $\ell S$-point DFT grid. Define the (oversampled) IDFT matrix
$\mathbf F_{\ell S}^{H}\in\mathbb C^{\ell S\times S}$ as $
\mathbf F_{\ell S}^{H}[n,s]
\triangleq \frac{1}{\sqrt{\ell S}}e^{\jmath 2\pi \frac{ns}{\ell S}}, n\in\{0,1,\ldots,\ell S-1\},\ \ s\in\mathcal{S}.$
Then, the useful (no-CP) time-domain block $\mathbf x^{a,(b)}\in\mathbb C^{\ell S}$ is
\begin{equation}
	\mathbf x^{a,(b)}=\mathbf F_{\ell S}^{H}\mathbf g^{a,(b)}.
	\label{eq:ifft_digital}
\end{equation}

Large peaks in the time-domain OFDM waveform drive the PA into saturation and cause nonlinear distortion and spectral regrowth. A common remedy is amplitude clipping, which caps the signal magnitude at a threshold $\chi$, but clipping itself introduces in-band distortion and worsens OOB leakage. We instead design the precoder so that the realized per-antenna waveform inherently satisfies the peak limit, requiring
\begin{equation}
	\| \mathbf{x}^{a,(b)} \|_{\infty} \leq \chi ,\qquad \forall a,b.
	\label{eq:clipping_constraint_digital}
\end{equation}
Since the CP is formed by copying the last $\ell N_{\mathrm{CP}}$ samples of the useful OFDM block, enforcing $|\mathbf x^{a,(b)}[n]|\le \chi$ for $0\le n\le \ell S-1$ automatically guarantees the same bound on the CP samples $n\in\{-\ell N_{\mathrm{CP}},\ldots,-1\}$.

\subsubsection{Spectral Mask Constraint}

We enforce a spectral mask to limit OOB emissions, with denser suppression near band edges. Let $F_{s,\ell}$ denote the sampling rate of the oversampled discrete-time waveform.
For spectral sampling, introduce a set of (possibly non-integer) locations
$\{\gamma_i\}_{i=0}^{M-1}\subset\mathbb R$ expressed in DFT-bin units on the $\ell S$ grid; the location $\gamma_i$
corresponds to the tangible baseband frequency $f_i=\frac{\gamma_i}{\ell S}F_{s,\ell}$ (Hz). Integer $\gamma_i$ coincide with
$\ell S$-point DFT-bin centers, whereas non-integer $\gamma_i$ evaluate the DFT between bins.

Using the rectangular window in \eqref{eq:timedomain_digital}, define the sampling matrix $\mathbf A\in\mathbb C^{M\times S}$ by
evaluating that formula at each $\gamma_i$, i.e.,

\begin{align}
	&	\mathbf A[i,s]
	=
	\frac{1}{\sqrt{\ell S}}
	\exp\!\left(\jmath\pi\frac{\gamma_i-s}{\ell S}\,(\ell N_{\mathrm{CP}}-\ell S+1)\right)\nonumber\\
	&\times
	\frac{\sin\!\left(\pi\frac{\gamma_i-s}{\ell S}(\ell S+\ell N_{\mathrm{CP}})\right)}
	{\sin\!\left(\pi\frac{\gamma_i-s}{\ell S}\right)},
	\qquad \gamma_i-s\notin \ell S\,\mathbb Z,
	\label{eq:A_closed}
\end{align}
and $\mathbf A[i,s]=L/\sqrt{\ell S}$ when $\gamma_i-s\in \ell S\,\mathbb Z$, with $L\triangleq \ell S+\ell N_{\mathrm{CP}}$ \cite{9214877,van2009sculpting}.

For a fixed antenna $a$, \eqref{eq:ofdm_sym_digital}--\eqref{eq:timedomain_digital} imply that the CP-inclusive spectrum evaluated at
the sampled locations satisfies
\begin{align}
		X^{a,(b)}(\gamma_i)&
	\triangleq 
	\hspace{-.3cm}	\sum_{n=-\ell N_{\mathrm{CP}}}^{\ell S-1}\hspace{-.4cm} x_{\text{CP}}^{a,(b)}[n]\;e^{-\jmath 2\pi \frac{\gamma_i}{\ell S}n}
	=
	\sum_{s\in\mathcal S}\hspace{-.1cm}\mathbf g^{a,(b)}[s]\;\mathbf A[i,s]\nonumber\\
	&	=
	\mathbf A[i,:]\mathbf g^{a,(b)},
	\nonumber
\end{align}
and a (single-symbol) periodogram-type PSD sample at $\gamma_i$ is
	\begin{equation}
		\widehat S_{x_{\text{CP}}^{a,(b)} x_{\text{CP}}^{a,(b)}}(\gamma_i)
		\;\triangleq\;
		\frac{1}{L\,F_{s,\ell}}\bigl|X^{a,(b)}(\gamma_i)\bigr|^2
		\;=\;
		\frac{1}{L\,F_{s,\ell}}\bigl|\mathbf A[i,:]\mathbf g^{a,(b)}\bigr|^2.\nonumber
	\end{equation}
We use the single-symbol periodogram $\widehat S_{x_{\text{CP}}^{a,(b)} x_{\text{CP}}^{a,(b)}}(\gamma)= \frac{1}{L F_{s,\ell}}|X^{a,(b)}(\gamma)|^2$ as a measure of spectral leakage. The spectral mask is enforced at a finite set of mask frequencies $\{f_1,\ldots,f_G\}$, distinct from the generic sampling locations $\{\gamma_i\}_{i=0}^{M-1}$: the latter can be densely chosen over a wide frequency range for PSD visualization, whereas $\{f_j\}$ is a design set used only for compliance. Following \cite{van2009sculpting}, the mask points are placed near band edges and interference regions, like an effective discrete surrogate for a continuous mask.

We map each physical mask frequency $f_j$ to its DFT-bin location on the $\ell S$ grid via
$\gamma_j \triangleq \frac{\ell S}{F_{s,\ell}}f_j$. Let $\mathbf A_n\in\mathbb C^{G\times S}$ denote the matrix obtained by
evaluating \eqref{eq:A_closed} at $\{\gamma_j\}_{j=1}^{G}$ (equivalently, selecting the corresponding rows of $\mathbf A$
when $\{\gamma_j\}\subset\{\gamma_i\}$). Then, $
X^{a,(b)}(\gamma_j)=\mathbf A_n[j,:]\mathbf g^{a,(b)}, j=1,\ldots,G.$
The mask constraints at antenna $a$ are imposed as
\begin{equation}
	\frac{\left|\mathbf A_n \mathbf g^{a,(b)}\right|^{\circ 2}}{L\,F_{s,\ell}} \preceq \frac{\mathbf r}{L\,F_{s,\ell}},
	\qquad \mathbf r=[\mathbf r[1],\ldots,\mathbf r[G]]^T, \forall b, \forall a,
	\label{eq:mask_constraint}
\end{equation}
with $\mathbf r[j] \triangleq L\,F_{s,\ell}\, S_{\text{max}}(f_j)$, and  $S_{\text{max}}(f_j)$ being the maximum allowable PSD. Equivalently, \eqref{eq:mask_constraint} designs $\{\mathbf V_k^{s,(b)}\}$ so that each $\mathbf g^{a,(b)}$ meets the prescribed spectral limits at $\{f_1,\ldots,f_G\}$.

\section{Problem Formulation}
\label{sec:probform}
We adopt a hybrid sum-MSE objective that reflects the asymmetric symbol knowledge: the transmitter knows the realized QAM tuple and designs $\{\mathbf t^{s,(b)},\mathbf g^{a,(b)},\mathbf x^{a,(b)}\}$ pointwise per $b$, whereas the receiver does not and therefore employs symbol-agnostic linear combiners $\{\mathbf U_k^s\}$. A compatible $\{\mathbf V_k^{s,(b)}\}$ is recovered from $\{\mathbf t^{s,(b)}\}$ via Proposition~\ref{prop:recover_V_from_t}.
With estimation error $\mathbf e_k^{s,(b)}\triangleq \mathbf U_k^{sH}\mathbf y_k^{s,(b)}-\boldsymbol{\omega}_k^{s,(b)}$, the problem is
	\begin{equation}\label{opt:main}
		\begin{aligned}
			\min_{\{\mathbf U_k^s\},\,\{\mathbf t^{s,(b)}\},\,\{\mathbf x^{a,(b)}\},\,\{\mathbf g^{a,(b)}\}}
			\quad
			& J \triangleq \tfrac{1}{B}\sum_{b,s,k}
			\mathbb E_{\mathbf n}\!\left[\|\mathbf e_k^{s,(b)}\|_2^2\right] \\
			\text{s.t.}\quad
			& \eqref{eq:powerbudget_digital},\ \eqref{eq:ga_def_digital},\ \eqref{eq:ifft_digital},\ \eqref{eq:clipping_constraint_digital},\ \eqref{eq:mask_constraint},
			\qquad \forall\,b.
		\end{aligned}
	\end{equation}
The constraints are imposed per realization $b$, while $J$ averages noise-MSE across the batch to match the receiver's symbol-agnostic operation.
\section{The Proposed BCD-Based Algorithm}
\label{sec:alg} 
Note that $J=\tfrac{1}{B}\sum_{b,s,k}\|\mathbf U_k^{sH}\mathbf H_k^s\mathbf t^{s,(b)}-\boldsymbol{\omega}_k^{s,(b)}\|_2^2+\sum_{s,k}\sigma_{\mathrm{noise},k}^{s\,2}\tr(\mathbf U_k^{sH}\mathbf U_k^s)$. At each BCD iteration, the $B$ transmit-side subproblems are independent across $b$ and solved in parallel via the ADMM of Section~\ref{subsec:admm_blocks_revised_u}; the resulting $\{\mathbf t^{s,(b)}\}_{b=1}^B$ form the sample-average covariances for the combiner update.
\subsection{The Transmit-Side Subproblem}
For fixed $\{\mathbf U_k^s\}$, the transmit-side part of $J$ equals $\tfrac{1}{B}\sum_{b,s,k}\|\mathbf B_k^s\mathbf t^{s,(b)}-\boldsymbol{\omega}_k^{s,(b)}\|_2^2+\mathrm{const}$, $\mathbf B_k^s\triangleq \mathbf U_k^{sH}\mathbf H_k^s$. Since cost and constraints decouple across $b$, the problem separates into $B$ pointwise subproblems, one per $b$:
\begin{equation}\label{opt:main-revised}
	\begin{aligned}
		\min_{\{\mathbf t^{s,(b)}\},\{\mathbf x^{a,(b)}\},\{\mathbf g^{a,(b)}\}}
		&\ \sum_{s,k}\|\mathbf B_k^s\mathbf t^{s,(b)}-\boldsymbol{\omega}_k^{s,(b)}\|_2^2\\
		\text{s.t.}\ & \eqref{eq:powerbudget_digital},\eqref{eq:ga_def_digital},\eqref{eq:ifft_digital},\eqref{eq:clipping_constraint_digital},\eqref{eq:mask_constraint},
	\end{aligned}
\end{equation}
handled by the four-block ADMM in Section~\ref{subsec:admm_blocks_revised_u}.

\subsection{A Four-Block ADMM Reformulation}
\label{subsec:admm_blocks_revised_u}

For fixed $\{\mathbf U_k^s\}$, the transmit-side subproblem is solved independently for each $b$, yielding $\mathbf t^{s,(b)}$; the $B$ instances may be solved in parallel. In what follows, we drop the superscript $(b)$ and present the solution for a single realization. Since $\mathbf g^a[s]=\mathbf t^s[a]$ for every $a\in\{1,\ldots,N_t\}$ and $s\in\mathcal S$, the inner problem may be expressed entirely through $\{\mathbf t^s\}$, $\{\mathbf g^a\}$, and $\{\mathbf x^a\}$. To handle the spectral-mask constraint, define
$\mathbf q^a \triangleq \mathbf A_n \mathbf g^a \in \mathbb C^{G}$ for $a=1,\ldots,N_t$. Hence, the inner problem is written in terms of
$\{\mathbf t^s\}$, $\{\mathbf g^a\}$, $\{\mathbf x^a\}$, and $\{\mathbf q^a\}$. Once $\{\mathbf t^s\}$ is available, a compatible $\{\mathbf V_k^{s,(b)}\}$ is reconstructed afterward; see Proposition~\ref{prop:recover_V_from_t}. We also introduce two quadratic regularization
terms with parameters $\eta_g>0$ and $\eta_t>0$. The stacked variables are
$\mathbf T \triangleq [\mathbf t^0 \cdots \mathbf t^{S-1}]\in\mathbb C^{N_t\times S}$,
$\mathbf W \triangleq [\mathbf g^1\cdots \mathbf g^{N_t}]^T\in\mathbb C^{N_t\times S}$,
$\mathbf X \triangleq [\mathbf x^1 \cdots \mathbf x^{N_t}]^T\in\mathbb C^{N_t\times \ell S}$, and
$\mathbf Q \triangleq [\mathbf q^1 \cdots \mathbf q^{N_t}]^T\in\mathbb C^{N_t\times G}$. By construction, $\mathbf{W}[a,:] = \mathbf{T}[a,:]$ for all $a$. Since $\mathbf{W}$ is arranged row-wise, we introduce the operators $\mathfrak{F}(\mathbf{W}) \triangleq \mathbf{W}(\mathbf{F}_{\ell S}^{H})^{T} \in \mathbb{C}^{N_t \times \ell S}$ and $\mathfrak{A}(\mathbf{W}) \triangleq \mathbf{W}\mathbf{A}_n^{T} \in \mathbb{C}^{N_t \times G}$.
	
	Next, define the closed convex sets
		\begin{align}
			\mathcal M &\triangleq \Big\{\mathbf Q\in\mathbb C^{N_t\times G}: |\mathbf Q[a,j]|^2\le \mathbf r[j],\ \forall a,j\Big\}, \nonumber\\
			\mathcal C &\triangleq \Big\{\mathbf X\in\mathbb C^{N_t\times \ell S}: |\mathbf X[a,n]|\le \chi,\ \forall a,n\Big\}, \nonumber\\
			\mathcal P &\triangleq \Big\{\mathbf T\in\mathbb C^{N_t\times S}: \|\mathbf T[:,s]\|_2^2\le P^s,\ \forall s\in\mathcal S\Big\}.\nonumber
		\end{align}
	For any closed convex set $\mathcal Z$, let $\delta_{\mathcal Z}(\mathbf Z)$ denote its indicator function, equal to $0$ if $\mathbf Z\in\mathcal Z$ and $+\infty$ otherwise.
With $\{\mathbf U_k^s\}$ fixed, the regularized inner problem for the current realization $\boldsymbol{\omega}$ becomes
	\begin{align}
		\min_{\mathbf Q,\mathbf X,\mathbf W,\mathbf T}&\quad
		\theta_1(\mathbf Q)+\theta_2(\mathbf X)+\theta_3(\mathbf W)+\theta_4(\mathbf T) \label{opt:inner4block}\\
		&\text{s.t.}\quad
		\mathbf Q-\mathfrak A(\mathbf W)=\mathbf 0,
		\mathbf X-\mathfrak F(\mathbf W)=\mathbf 0,
		\mathbf W-\mathbf T=\mathbf 0,\nonumber
	\end{align}
where $\theta_1(\mathbf Q)\triangleq \delta_{\mathcal M}(\mathbf Q)$, $\theta_2(\mathbf X)\triangleq \delta_{\mathcal C}(\mathbf X)$, $\theta_3(\mathbf W)\triangleq \frac{\eta_g}{2}\|\mathbf W\|_F^2$, and $\theta_4(\mathbf T)\triangleq\sum_{s\in\mathcal S}\sum_{k=1}^K \|\mathbf B_k^s\mathbf T[:,s]-\boldsymbol\omega_k^s\|_2^2+\delta_{\mathcal P}(\mathbf T)+\frac{\eta_t}{2}\|\mathbf T\|_F^2$.
	
	Let $\langle \mathbf A,\mathbf B\rangle \triangleq \Re\{\tr(\mathbf A^H\mathbf B)\}$. Introducing the dual matrices
	$\boldsymbol\Lambda_q\in\mathbb C^{N_t\times G}$, $\boldsymbol\Lambda_x\in\mathbb C^{N_t\times \ell S}$, and
	$\boldsymbol\Lambda_g\in\mathbb C^{N_t\times S}$, the augmented Lagrangian is
		\begin{align}
			&\mathcal L(\mathbf Q,\mathbf X,\mathbf W,\mathbf T,\boldsymbol\Lambda_q,\boldsymbol\Lambda_x,\boldsymbol\Lambda_g)=
			\theta_1(\mathbf Q)+\theta_2(\mathbf X)+\theta_3(\mathbf W) + 
			\nonumber\\
			&\hspace{-.2cm}\theta_4(\mathbf T)
			+\langle \boldsymbol\Lambda_q,\mathbf Q-\mathfrak A(\mathbf W)\rangle
			+\rho/2\|\mathbf Q-\mathfrak A(\mathbf W)\|_F^2 +\langle \boldsymbol\Lambda_x,\mathbf X-\mathfrak F(\mathbf W)\rangle
			\nonumber\\
			&
			+\frac{\rho}{2}\|\mathbf X-\mathfrak F(\mathbf W)\|_F^2
			+\langle \boldsymbol\Lambda_g,\mathbf W-\mathbf T\rangle
			+\frac{\rho}{2}\|\mathbf W-\mathbf T\|_F^2 .
			\label{eq:lagran_4block}
		\end{align}
	Accordingly, at iteration $\tau+1$, the cyclic ADMM steps are
	\begin{equation}\label{eq:admm_4block_updates}
		\begin{split}
			&\mathbf Q^{\tau+1}
			=
			\arg\min_{\mathbf Q}\ \mathcal L\!\left(
			\mathbf Q,\mathbf X^\tau,\mathbf W^\tau,\mathbf T^\tau,
			\boldsymbol\Lambda_q^\tau,\boldsymbol\Lambda_x^\tau,\boldsymbol\Lambda_g^\tau
			\right),\\
			&\mathbf X^{\tau+1}
			=
			\arg\min_{\mathbf X}\ \mathcal L\!\left(
			\mathbf Q^{\tau+1},\mathbf X,\mathbf W^\tau,\mathbf T^\tau,
			\boldsymbol\Lambda_q^\tau,\boldsymbol\Lambda_x^\tau,\boldsymbol\Lambda_g^\tau
			\right),\\
			&\mathbf W^{\tau+1}
			=
			\arg\min_{\mathbf W}\ \mathcal L\!\left(
			\mathbf Q^{\tau+1},\mathbf X^{\tau+1},\mathbf W,\mathbf T^\tau,
			\boldsymbol\Lambda_q^\tau,\boldsymbol\Lambda_x^\tau,\boldsymbol\Lambda_g^\tau
			\right),\\
			&\mathbf T^{\tau+1}
			=
			\arg\min_{\mathbf T}\ \mathcal L\!\left(
			\mathbf Q^{\tau+1},\mathbf X^{\tau+1},\mathbf W^{\tau+1},\mathbf T,
			\boldsymbol\Lambda_q^\tau,\boldsymbol\Lambda_x^\tau,\boldsymbol\Lambda_g^\tau
			\right),\\
			&\boldsymbol\Lambda_q^{\tau+1}
			=
			\boldsymbol\Lambda_q^\tau
			+\rho\big(\mathbf Q^{\tau+1}-\mathfrak A(\mathbf W^{\tau+1})\big),\\
			&\boldsymbol\Lambda_x^{\tau+1}
			=
			\boldsymbol\Lambda_x^\tau
			+\rho\big(\mathbf X^{\tau+1}-\mathfrak F(\mathbf W^{\tau+1})\big),\\
			&\boldsymbol\Lambda_g^{\tau+1}
			=
			\boldsymbol\Lambda_g^\tau
			+\rho\big(\mathbf W^{\tau+1}-\mathbf T^{\tau+1}\big).
		\end{split}
		\raisetag{-8pt}
	\end{equation}
	\subsubsection{Update of $\mathbf Q$}
	\label{subsec:q_solution_u}
	
	The $\mathbf Q$-subproblem reduces to Euclidean projection onto $\mathcal M$ as 
	\begin{align}
\mathbf Q^{\tau+1}
=
\operatorname{Proj}_{\mathcal M}\!\left(\mathfrak A(\mathbf W^\tau)-\frac{1}{\rho}\boldsymbol\Lambda_q^\tau\right).\nonumber
	\end{align}
	Equivalently, for every antenna $a$ and frequency sample $j$,
	\begin{align}
	&\mathbf q^{a,\tau+1}[j]
	=
	\min\!\Big(
	1,\frac{\sqrt{\mathbf r[j]}}{\left|\mathbf A_n[j,:]\mathbf g^{a,\tau}-\boldsymbol\Lambda_q^\tau[a,j]/\rho\right|}
	\Big)
	\nonumber\\
	&\qquad\times
	\left(\mathbf A_n[j,:]\mathbf g^{a,\tau}-\boldsymbol\Lambda_q^\tau[a,j]/\rho\right), j=1,\ldots,G,
	\label{eq:q_closed_u}
\end{align}
	with the convention that $\mathbf q^{a,\tau+1}[j]=0$ whenever the term inside parentheses is zero.
	
	\subsubsection{Update of $\mathbf X$}
	\label{subsec:x_solution_u_new}
	
	The $\mathbf X$-subproblem is likewise a Euclidean projection, now onto $\mathcal C$:
	\[
\mathbf X^{\tau+1}
=
\operatorname{Proj}_{\mathcal C}\!\left(\mathfrak F(\mathbf W^\tau)-\frac{1}{\rho}\boldsymbol\Lambda_x^\tau\right).
\]
	In elementwise form, for each antenna $a$ and sample index $n=0,\ldots,\ell S-1$,
	\begin{align}
	\mathbf x^{a,\tau+1}[n]&
	=
	\min\!\Big(
	1,\frac{\chi}{\left|\mathbf F_{\ell S}^{H}[n,:]\mathbf g^{a,\tau}-\boldsymbol\Lambda_x^\tau[a,n]/\rho\right|}
	\Big)
	\nonumber\\
	&\qquad\times
	\left(\mathbf F_{\ell S}^{H}[n,:]\mathbf g^{a,\tau}-\boldsymbol\Lambda_x^\tau[a,n]/\rho\right).
	\label{eq:x_closed_u_new}
\end{align}
	
	\subsubsection{Update of $\mathbf W$}
	\label{subsec:g_solution_u_new}
	
	The minimization with respect to $\mathbf W$ is unconstrained, strongly convex, and separable across antennas. For each
	$a$, the first-order optimality condition yields
	\begin{align}
		&\hspace{-.3cm}\Big((\eta_g+\rho)\mathbf I_S+\rho\mathbf A_n^H\mathbf A_n+\rho\mathbf F_{\ell S}\mathbf F_{\ell S}^{H}\Big)\mathbf g^{a,\tau+1}
		\nonumber\\
		&\hspace{-.3cm}=
		\mathbf A_n^H\!\left(\rho\mathbf q^{a,\tau+1}+(\boldsymbol\Lambda_q^\tau[a,:])^T\right)
		+\mathbf F_{\ell S}\!\left(\rho\mathbf x^{a,\tau+1}+(\boldsymbol\Lambda_x^\tau[a,:])^T\right)
		\nonumber\\
		&\hspace{-.3cm}+\rho(\mathbf T^\tau[a,:])^T-(\boldsymbol\Lambda_g^\tau[a,:])^T.
		\label{eq:g_normal_per_antenna_raw}
	\end{align}
	Using $\mathbf F_{\ell S}\mathbf F_{\ell S}^{H}=\mathbf I_S$, this simplifies to
	\begin{align}
	&\mathbf M_g\,\mathbf g^{a,\tau+1}
	=
	\mathbf A_n^H\!\left(\rho\mathbf q^{a,\tau+1}+(\boldsymbol\Lambda_q^\tau[a,:])^T\right)
	\label{eq:g_normal_per_antenna}\\
	&+\mathbf F_{\ell S}\!\big(\rho\mathbf x^{a,\tau+1}+(\boldsymbol\Lambda_x^\tau[a,:])^T\big)+\rho(\mathbf T^\tau[a,:])^T-(\boldsymbol\Lambda_g^\tau[a,:])^T,
	\nonumber
\end{align}
	where $\mathbf M_g \triangleq (\eta_g+2\rho)\mathbf I_S+\rho\mathbf A_n^H\mathbf A_n
	\in\mathbb C^{S\times S}.$
	
	In practice, one solves \eqref{eq:g_normal_per_antenna} as a linear system instead of explicitly forming $\mathbf M_g^{-1}$. When
	$\rho$ remains fixed, $\mathbf M_g$ is identical for all antennas and across all ADMM iterations, so a single
	factorization can be cached and reused. In addition, if $G\ll S$, the Woodbury matrix identity gives
		\begin{align}
			&\mathbf M_g^{-1}
			=
			\alpha^{-1}\mathbf I_S
			-\frac{\rho}{\alpha^2}\mathbf A_n^H
			\left(\mathbf I_G+\frac{\rho}{\alpha}\mathbf A_n\mathbf A_n^H\right)^{-1}
			\mathbf A_n,
			\alpha\triangleq \eta_g+2\rho,\nonumber
		\end{align}
	which reduces the computation to solving a $G\times G$ system. Otherwise, one may employ a cached Cholesky factorization of
	$\mathbf M_g$ or a preconditioned conjugate-gradient solve.
	
	\subsubsection{Update of $\mathbf T$}
	\label{subsec:t_solution_u_new}
	
	The $\mathbf T$-subproblem is strongly convex and separates over $s\in\mathcal S$. For each $s\in\mathcal S$,
	\begin{align}
	&\min_{\mathbf t^s}\quad
	\sum_{k=1}^{K}\|\mathbf B_k^s\mathbf t^s-\boldsymbol\omega_k^s\|_2^2
	+\frac{\eta_t}{2}\|\mathbf t^s\|_2^2
	\nonumber\\
	&\qquad+\Re\!\left\{(\boldsymbol\Lambda_g^\tau[:,s])^H(\mathbf W^{\tau+1}[:,s]-\mathbf t^s)\right\}
	\nonumber\\
	&\qquad
	+\frac{\rho}{2}\|\mathbf W^{\tau+1}[:,s]-\mathbf t^s\|_2^2
	\qquad
	\text{s.t.}\quad
	\|\mathbf t^s\|_2^2\le P^s .
	\label{opt:t_sub_u_new}
\end{align}
	Whenever the minimizer of the objective function satisfies $\|\mathbf t^s\|_2^2\le P^s$, it is also the optimizer of
	\eqref{opt:t_sub_u_new}. Otherwise, the power constraint is active. Let $\mu^s\ge 0$ denote the Lagrange multiplier and set
	\[
\mathbf M_t^s(\mu^s)\triangleq
2\sum_{k=1}^{K}(\mathbf B_k^s)^H\mathbf B_k^s+(\eta_t+\rho+2\mu^s)\mathbf I_{N_t}
\in\mathbb C^{N_t\times N_t}.
\]
	Then, $\mathbf t^s(\mu^s)
			=
			\big(\mathbf M_t^s(\mu^s)\big)^{-1}
			\big(
			2\sum_{k=1}^{K}(\mathbf B_k^s)^H\boldsymbol\omega_k^s
			+\rho\mathbf W^{\tau+1}[:,s]
			+\boldsymbol\Lambda_g^\tau[:,s]
			\big),$
	with $\mu^s>0$ such that $\|\mathbf t^s(\mu^s)\|_2^2=P^s$.
	
	The updates for $\mathbf Q$, $\mathbf X$, and $\mathbf W$ are separable across antennas, whereas the update for $\mathbf T$ separates across
	subcarriers.

\begin{proposition}
	\label{prop:recover_V_from_t}
	For each $s\in\mathcal S$, let $\bar{\boldsymbol\omega}^{\,s}\triangleq\big[(\boldsymbol\omega_1^s)^T~\cdots~(\boldsymbol\omega_K^s)^T\big]^T$. A feasible recovery is any $\{\mathbf V_k^s\}_{k=1}^K$ with $\sum_{k=1}^K \mathbf V_k^s[a,:]\boldsymbol\omega_k^s=\mathbf t^s[a]$, $\forall a$. If $\|\bar{\boldsymbol\omega}^{\,s}\|_2^2>0$, the minimum-Frobenius-norm solution is
	\[
	\mathbf V_k^s[a,:]
	=
	\frac{\mathbf t^s[a]}{\sum_{j=1}^K\|\boldsymbol\omega_j^s\|_2^2}\,(\boldsymbol\omega_k^s)^H,
	\qquad \forall a,\ k.
	\]
\end{proposition}
	
	\begin{proof}
		Fix $s\in\mathcal S$ and define $\bar{\mathbf v}^{\,s}[a,:]\triangleq [\mathbf V_1^s[a,:]\ \cdots\ \mathbf V_K^s[a,:]]$. Then,
		$\bar{\mathbf v}^{\,s}[a,:]\bar{\boldsymbol\omega}^{\,s}=\sum_{k=1}^K \mathbf V_k^s[a,:]\boldsymbol\omega_k^s=\mathbf t^s[a]$, so the
		recovery problem separates row by row. Moreover,
		$\sum\nolimits_{k=1}^K\|\mathbf V_k^s\|_F^2=\sum\nolimits_{a=1}^{N_t}\|\bar{\mathbf v}^{\,s}[a,:]\|_2^2$, and therefore the minimum-Frobenius-norm recovery is
		obtained by minimizing each row norm under the constraint $\bar{\mathbf v}^{\,s}[a,:]\bar{\boldsymbol\omega}^{\,s}=\mathbf t^s[a]$. If
		$\|\bar{\boldsymbol\omega}^{\,s}\|_2^2>0$, the minimum-norm solution is
		$\bar{\mathbf v}^{\,s}[a,:]=\mathbf t^s[a](\bar{\boldsymbol\omega}^{\,s})^H/\|\bar{\boldsymbol\omega}^{\,s}\|_2^2$, which gives
		$\mathbf V_k^s[a,:]=\mathbf t^s[a](\boldsymbol\omega_k^s)^H/\sum_{j=1}^K\|\boldsymbol\omega_j^s\|_2^2$.
	\end{proof}
	
\begin{proposition}
	\label{prop:admm_convergence}
	Assume $\eta_g,\eta_t>0$.
	Let $\mathcal R(\mathbf Z)\triangleq
	\begin{bmatrix}\Re\{\operatorname{vec}(\mathbf Z)\}\\
		\Im\{\operatorname{vec}(\mathbf Z)\}\end{bmatrix}$,
	and define $\boldsymbol{\beta}_1=\mathcal R(\mathbf Q)$,
	$\boldsymbol{\beta}_2=\mathcal R(\mathbf X)$,
	$\boldsymbol{\beta}_3=\mathcal R(\mathbf W)$, and
	$\boldsymbol{\beta}_4=\mathcal R(\mathbf T)$.
	Let $\mathbf A_{\mathcal R}$ and $\mathbf F_{\mathcal R}$ be the unique
	real matrices satisfying
	$\mathcal R(\mathfrak A(\mathbf W))=\mathbf A_{\mathcal R}\boldsymbol{\beta}_3$
	and
	$\mathcal R(\mathfrak F(\mathbf W))=\mathbf F_{\mathcal R}\boldsymbol{\beta}_3$.
	If $\rho$ satisfies \cite[Eq.~(3.37)]{tao2018convergence} for the
	real-valued problem
		\[
		\begin{aligned}
		\hspace{-.2cm}	\min\quad
			&\hspace{-.2cm}\theta_1(\mathcal R^{-1}(\boldsymbol{\beta}_1))
			+\theta_2(\mathcal R^{-1}(\boldsymbol{\beta}_2))
			+\theta_3(\mathcal R^{-1}(\boldsymbol{\beta}_3))
			+\theta_4(\mathcal R^{-1}(\boldsymbol{\beta}_4))\\
			\mathrm{s.t.}\quad
			&\left[\begin{smallmatrix}
				\mathbf I_{2N_tG} & \mathbf 0 & -\mathbf A_{\mathcal R} & \mathbf 0\\
				\mathbf 0 & \mathbf I_{2N_t\ell S} & -\mathbf F_{\mathcal R} & \mathbf 0\\
				\mathbf 0 & \mathbf 0 & \mathbf I_{2N_tS} & -\mathbf I_{2N_tS}
			\end{smallmatrix}\right]
			\left[\begin{smallmatrix}
				\boldsymbol{\beta}_1\\\boldsymbol{\beta}_2\\\boldsymbol{\beta}_3\\\boldsymbol{\beta}_4
			\end{smallmatrix}\right]=\mathbf 0,
		\end{aligned}
		\]
	then the ADMM iterates~\eqref{eq:admm_4block_updates} converge to a
	primal--dual solution of~\eqref{opt:inner4block}. The primal limit
	$(\mathbf Q^\star,\mathbf X^\star,\mathbf W^\star,\mathbf T^\star)$
	is globally optimal and satisfies the KKT conditions
	of~\eqref{opt:inner4block}.
	
	Moreover, if $\{(\eta_g^\nu,\eta_t^\nu)\}_\nu$ is any sequence
	with $(\eta_g^\nu,\eta_t^\nu)\to(0,0)$ and
	$(\eta_g^\nu,\eta_t^\nu)>\mathbf 0$, and
	$(\mathbf Q^\nu,\mathbf X^\nu,\mathbf W^\nu,\mathbf T^\nu)$ is a
	global optimum of~\eqref{opt:inner4block} with parameters
	$(\eta_g^\nu,\eta_t^\nu)$, then every accumulation point of
	$\{(\mathbf Q^\nu,\mathbf X^\nu,\mathbf W^\nu,\mathbf T^\nu)\}_\nu$
	is globally optimal for~\eqref{opt:inner4block} in the limit $(\eta_g,\eta_t)\to\mathbf 0$.
\end{proposition}
	
	\begin{proof}
		Since $\mathcal R$ is a linear bijection, \eqref{opt:inner4block} and its realified counterpart are
		equivalent. We verify the hypotheses of
		\cite[Thm.~3.1]{tao2018convergence} for the realified four-block problem.
		
		\emph{Convexity.}
		$\theta_1=\delta_{\mathcal M}$ and $\theta_2=\delta_{\mathcal C}$ are closed
		proper convex as indicators of nonempty closed convex sets, while
		$\eta_g,\eta_t>0$ make $\theta_3$ and $\theta_4$ strongly convex, satisfying
		the $4-2=2$ strong-convexity requirement.
		
		\emph{Full column rank.}
		Blocks~1, 2, and~4 of the constraint matrix are identity. For block~3,
		$\mathbf A_{\mathcal R}^T\mathbf A_{\mathcal R}
		+\mathbf F_{\mathcal R}^T\mathbf F_{\mathcal R}
		+\mathbf I_{2N_tS}\succ\mathbf 0$.
		
		\emph{Relative interior feasibility.}
		The origin is equality-feasible and, since $\mathbf r[j],\chi,P^s>0$, lies in the
		interior of $\mathcal M$, $\mathcal C$, and $\mathcal P$, confirming the
		assumptions of~\cite{tao2018convergence}.
		
		\emph{Existence.}
		$\mathbf T\in\mathcal P$ is bounded; the equalities $\mathbf W=\mathbf T$,
		$\mathbf Q=\mathfrak A(\mathbf W)$, $\mathbf X=\mathfrak F(\mathbf W)$
		then bound all variables. The feasible set is nonempty and compact, and
		the objective is proper and lower semicontinuous, so a minimum exists.
		
		With $\rho$ in the admissible range,
		\cite[Thm.~3.1]{tao2018convergence} gives convergence; mapping back
		through $\mathcal R^{-1}$ proves the complex-domain claim, and convexity
		of~\eqref{opt:inner4block} yields global optimality and KKT satisfaction.
		
\emph{Vanishing regularization.}
Let $\boldsymbol{\zeta}\triangleq(\mathbf Q,\mathbf X,\mathbf W,\mathbf T)$,
$\mathcal F$ be the compact feasible set of~\eqref{opt:inner4block}, and
$F_0$ be the unregularized objective. Since $\mathcal F$ is compact, the
sequence $\{\boldsymbol{\zeta}^\nu\}$ of regularized optima has an accumulation
point $\bar{\boldsymbol{\zeta}}\in\mathcal F$ along some subsequence $\nu'$.
Optimality of $\boldsymbol{\zeta}^{\nu'}$ for the regularized problem yields
$F_0(\boldsymbol{\zeta}^{\nu'})+R^{\nu'}(\boldsymbol{\zeta}^{\nu'})\le F_0(\boldsymbol{\zeta})+R^{\nu'}(\boldsymbol{\zeta})$
for every $\boldsymbol{\zeta}\in\mathcal F$, with
$R^{\nu'}(\boldsymbol{\zeta})\triangleq\tfrac{\eta_g^{\nu'}}{2}\|\mathbf W\|_F^2
+\tfrac{\eta_t^{\nu'}}{2}\|\mathbf T\|_F^2$. Boundedness of $\mathcal F$
implies $R^{\nu'}(\boldsymbol{\zeta}^{\nu'})\to 0$ and $R^{\nu'}(\boldsymbol{\zeta})\to 0$,
so passing to the limit and invoking lower semicontinuity of $F_0$ gives $
	F_0(\bar{\boldsymbol{\zeta}})
	\le \liminf_{\nu'}F_0(\boldsymbol{\zeta}^{\nu'})
	\le \limsup_{\nu'}F_0(\boldsymbol{\zeta}^{\nu'})
	\le F_0(\boldsymbol{\zeta}),
	\forall\,\boldsymbol{\zeta}\in\mathcal F.$
Thus, $\bar{\boldsymbol{\zeta}}$ is globally optimal
in the limit $(\eta_g,\eta_t)\to\mathbf 0$.
	\end{proof}
	
\subsection{The Subproblem with Respect to $\mathbf U_k^s$}
For fixed transmit-side variables, the $\mathbf U_k^s$-subproblem is unconstrained and quadratic. The first order optimality condition yields the linear MMSE combiner
\begin{equation}
	\mathbf U_k^s=\big(\mathbf H_k^s\mathbf R_{t^st^s}(\mathbf H_k^s)^H+\sigma_{\mathrm{noise},k}^{s\,2}\mathbf I_{N_r}\big)^{-1}\mathbf H_k^s\mathbf R_{t^s\omega_k^s},
	\label{eq:U}
\end{equation}
where $\mathbf R_{t^st^s}\triangleq \frac{1}{B}\sum_{b}\mathbf t^{s,(b)}(\mathbf t^{s,(b)})^H$ and $\mathbf R_{t^s\omega_k^s}\triangleq \frac{1}{B}\sum_{b}\mathbf t^{s,(b)}(\boldsymbol{\omega}_k^{s,(b)})^H$, with each $\mathbf t^{s,(b)}$ produced by the inner ADMM of Section~\ref{subsec:admm_blocks_revised_u}. When the precoder is symbol-independent, $\mathbf R_{t^st^s}=\sum_{j=1}^{K}\mathbf V_j^s(\mathbf V_j^s)^H$ and $\mathbf R_{t^s\omega_k^s}=\mathbf V_k^s$, which recovers the conventional linear MMSE combiner. The symbol decision is $\mathcal Q(\mathbf U_k^{sH}\mathbf y_k^{s,(b)})$, with $\mathcal Q$ the componentwise QAM projection.
	
	\begin{figure*}[t!]
		\centering
		\subfloat[\label{fig:conv_bcd}]{
			\includegraphics[width=0.22\textwidth]{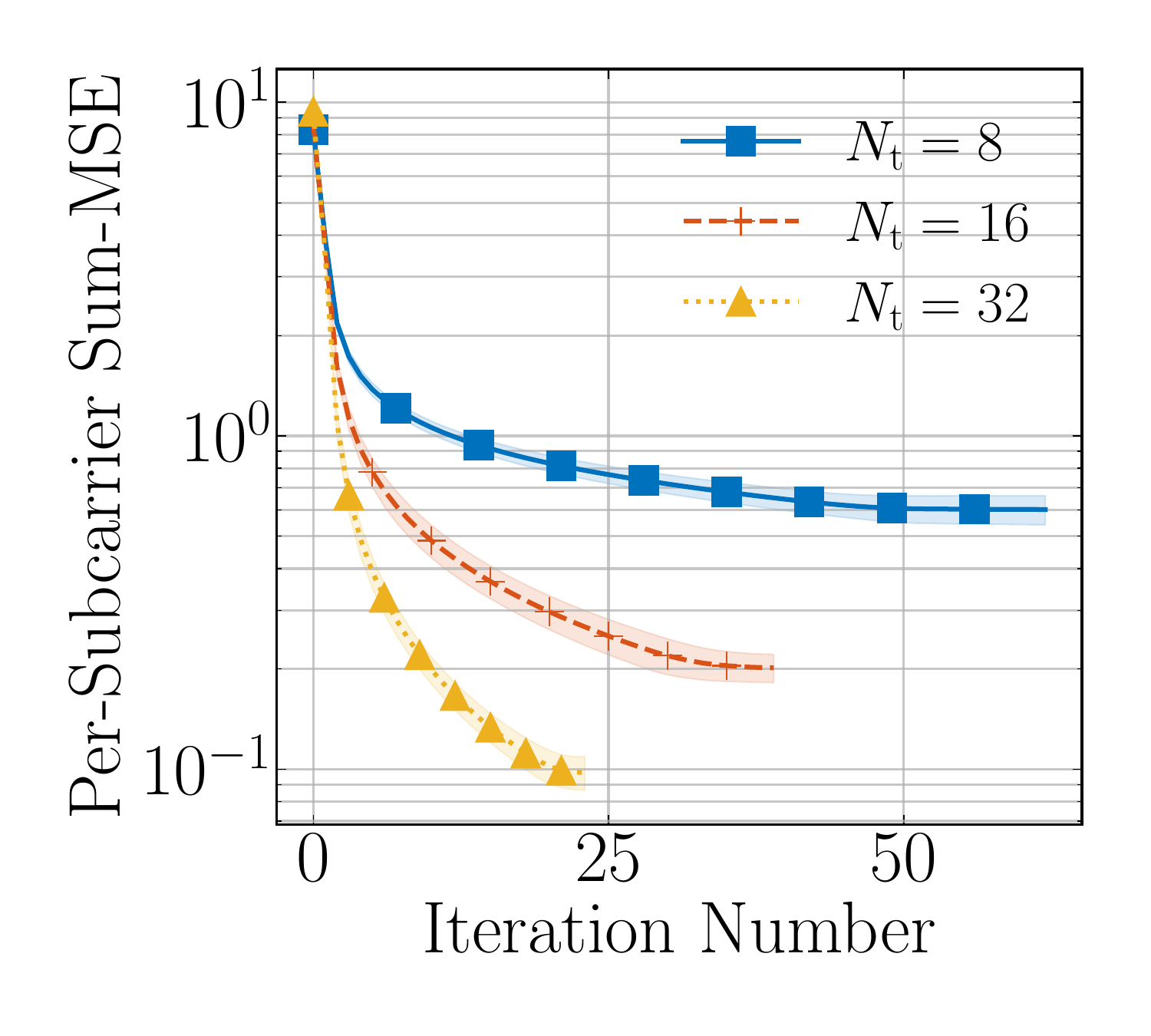}
		}\hfil
		\subfloat[\label{fig:rate_power_32}]{
			\includegraphics[width=0.22\textwidth]{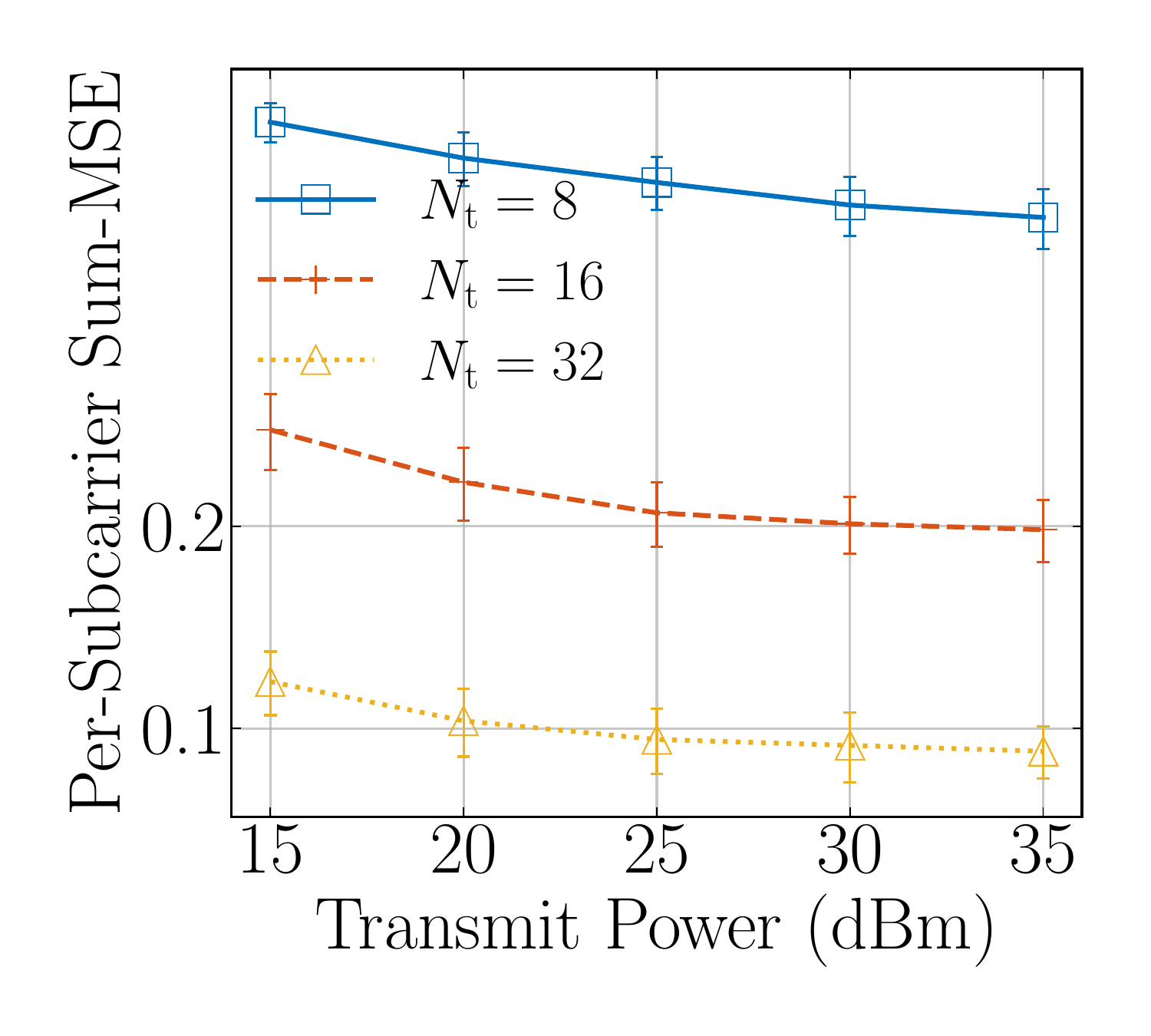}
		}\hfil
		\subfloat[\label{fig:rate_power_64}]{
			\includegraphics[width=0.22\textwidth]{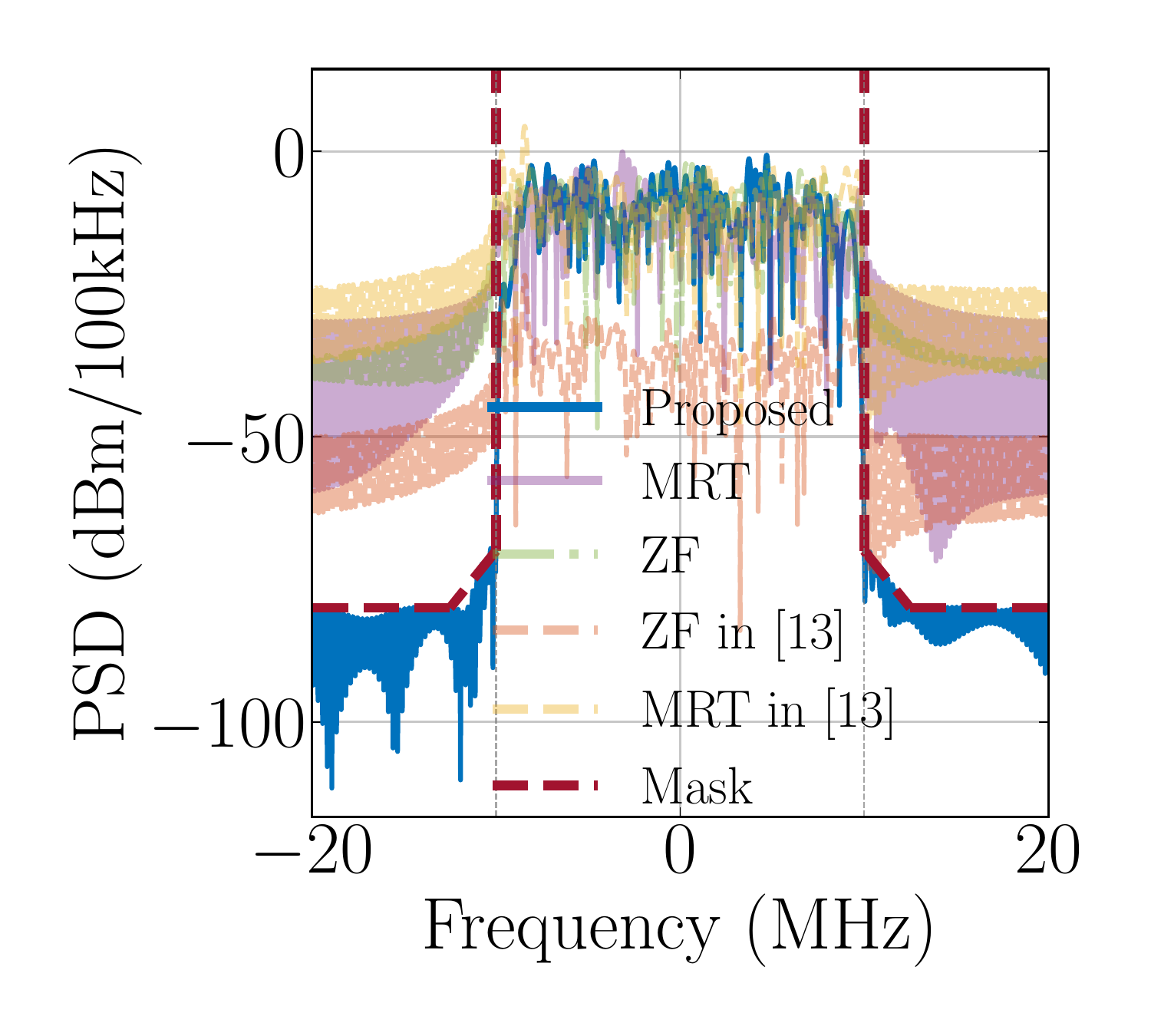}
		}
		\subfloat[\label{fig:mse_power_64}]{
			\includegraphics[width=0.22\textwidth]{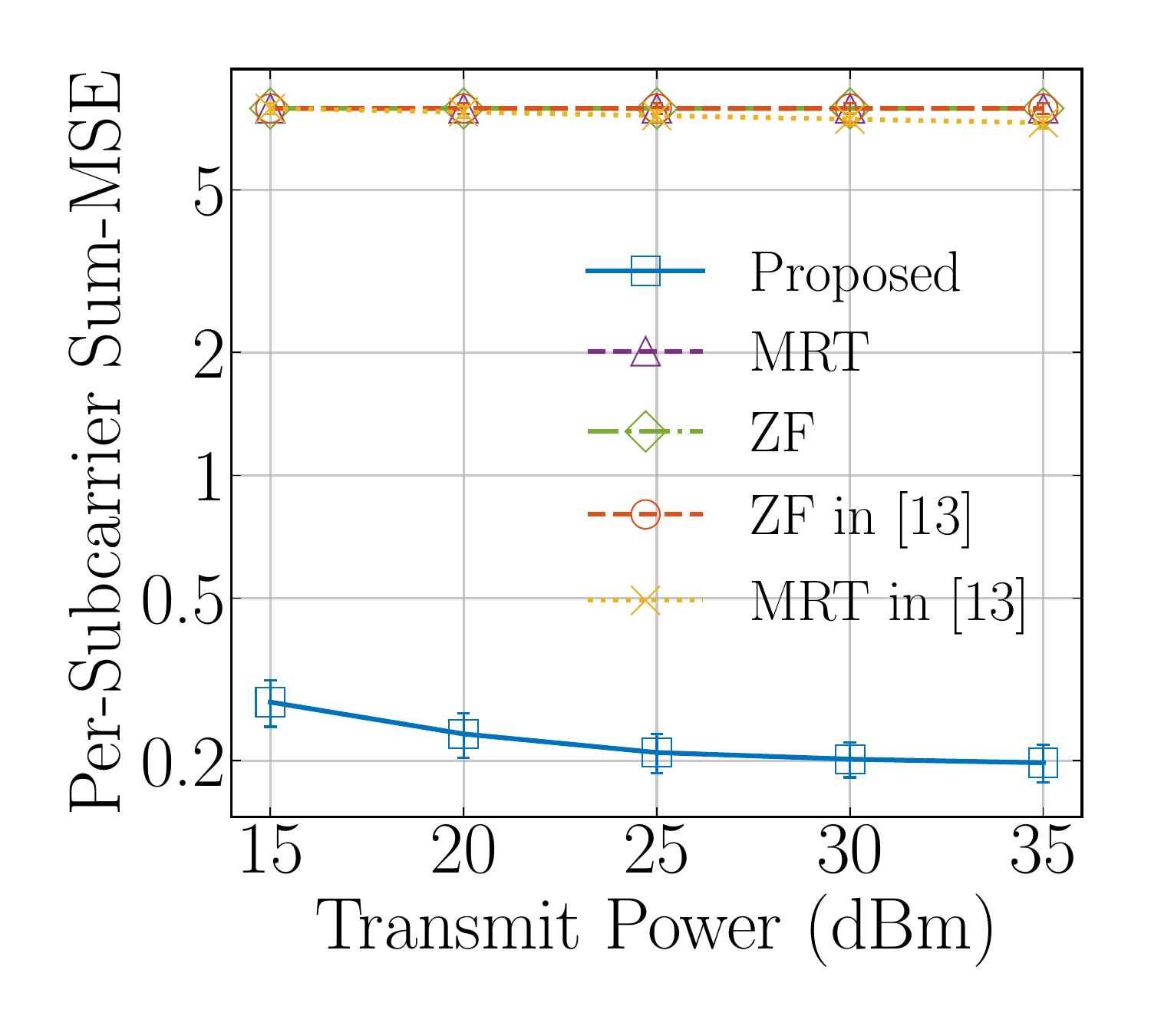}
		}
		\caption{Summary of convergence and performance metrics for the proposed scheme against benchmark methods.}
		\label{fig:summary_panels}
		
	\end{figure*}

	\subsection{The Proposed BCD-based Precoding/Combining Approach}

In the proposed method, the combiner block is updated using the symbol distribution, whereas the transmit-side block is solved pointwise for the realized symbol tuple.

\begin{proposition}
	\label{prop:bcd_stationarity}
	Assume $\eta_g,\eta_t>0$ and the hypotheses of Proposition~\ref{prop:admm_convergence} hold for every $b\in\{1,\ldots,B\}$. Let $J_{\eta_g,\eta_t}$ denote the objective of~\eqref{opt:main} augmented by the averaged per-realization regularizers $\frac{1}{B}\sum_{b=1}^{B}\left(\frac{\eta_g}{2}\|\mathbf W^{(b)}\|_F^2+\frac{\eta_t}{2}\|\mathbf T^{(b)}\|_F^2\right)$, and let $\mathcal F$ collect the stacked per-realization transmit quadruples satisfying the affine equalities and $\mathcal M,\mathcal C,\mathcal P$. For the idealized BCD alternating between~\eqref{eq:U} and the exact solution of~\eqref{opt:inner4block} per realization, indexed by $\alpha$: \emph{(i)}~$\{J_{\eta_g,\eta_t}^\alpha\}$ is non-increasing and convergent; \emph{(ii)}~every accumulation point is stationary for $\min_{\{\mathbf U_k^s\},\,\mathcal F}J_{\eta_g,\eta_t}$; \emph{(iii)}~for a separate sequence $(\eta_g^\nu,\eta_t^\nu)\to\mathbf 0$, every accumulation point of the corresponding stationary points is stationary for~\eqref{opt:main} in the limit $(\eta_g,\eta_t)\to\mathbf 0$.
\end{proposition}

\begin{proof}
Stacking the $B$ per-realization transmit quadruples with $\{\mathbf U_k^s\}$ gives a finite-dimensional two-block problem (Block~1: $\{\mathbf U_k^s\}$, unconstrained; Block~2: $\mathcal F$). Constraints decouple across $b$, so pointwise minimization yields the exact Block-2 minimizer (Proposition~\ref{prop:admm_convergence}); \eqref{eq:U} exactly minimizes Block~1. Hence, $0\le J_{\eta_g,\eta_t}^{\alpha+1}\le J_{\eta_g,\eta_t}^\alpha$, giving~\emph{(i)}. Block~1 is strongly convex with modulus at least $\sigma_{\min}^2\triangleq\min_{k,s}\sigma_{\mathrm{noise},k}^{s\,2}>0$; Block~2, after eliminating $\mathbf W^{(b)}=\mathbf T^{(b)}$, $\mathbf X^{(b)}=\mathfrak F(\mathbf T^{(b)})$, and $\mathbf Q^{(b)}=\mathfrak A(\mathbf T^{(b)})$, reduces to a program in $\mathbf T^{(b)}$ with strong-convexity modulus at least $(\eta_g+\eta_t)/B$ and hence has a unique minimizer. There exist $c_U,c_Z>0$ such that, with $J_{\eta_g,\eta_t}^{\alpha+\frac12}$ denoting the value after the Block-1 update, $J_{\eta_g,\eta_t}^\alpha-J_{\eta_g,\eta_t}^{\alpha+\frac12}\ge c_U\|\Delta\mathbf U^\alpha\|^2$ and $J_{\eta_g,\eta_t}^{\alpha+\frac12}-J_{\eta_g,\eta_t}^{\alpha+1}\ge c_Z\|\Delta\mathbf T^\alpha\|^2$. Telescoping yields $\|\Delta\mathbf U^\alpha\|,\|\Delta\mathbf T^\alpha\|\to 0$; the equalities propagate this to $\mathbf W^{(b)},\mathbf Q^{(b)},\mathbf X^{(b)}$. Power feasibility bounds $\mathbf T^{(b),\alpha}$; the equalities and~\eqref{eq:U} bound the rest, so all iterates lie in a compact set. Each block minimizer is a continuous function of the opposing block (Block~1 by continuity of~\eqref{eq:U}; Block~2 by Berge's theorem and uniqueness from strong convexity). Vanishing steps and continuity imply that any accumulation point is a fixed point of the two-block map, hence blockwise optimal: with $\mathbf Z \triangleq \{(\mathbf Q^{(b)},\mathbf X^{(b)},\mathbf W^{(b)},\mathbf T^{(b)})\}_{b=1}^{B}$ and $\bar{\mathbf Z}$ denoting any accumulation point of $\{\mathbf Z^\alpha\}$, one has $\nabla_{\mathbf U}J_{\eta_g,\eta_t}(\bar{\mathbf U},\bar{\mathbf Z})=\mathbf 0$ and $-\nabla_{\mathbf Z}J_{\eta_g,\eta_t}(\bar{\mathbf U},\bar{\mathbf Z})\in N_{\mathcal F}(\bar{\mathbf Z})$, (where $N_{\mathcal F}(\bar{\mathbf Z})$ denotes the normal cone to the feasible set $\mathcal F$ at $\bar{\mathbf Z}$), establishing~\emph{(ii)}. For~\emph{(iii)}, $\mathbf T^{(b),\nu}\in\mathcal P$ and~\eqref{eq:U} bound $\{(\mathbf U^\nu,\mathbf Z^\nu)\}$, so accumulation points exist. Along any convergent subsequence, $\nabla_{\mathbf U}J_{\eta_g^\nu,\eta_t^\nu}(\mathbf U^\nu,\mathbf Z^\nu)=\mathbf 0$ passes to the limit by continuity. The Block-2 inclusion differs from its limiting form only by $\tfrac{\eta_g^\nu}{B}\mathbf W^{(b),\nu}$ and $\tfrac{\eta_t^\nu}{B}\mathbf T^{(b),\nu}$, which are bounded and vanish; outer semicontinuity of $N_{\mathcal F}$ passes the inclusion to the limit. Therefore, every accumulation point is stationary for the equivalent lifted problem in the limit.
\end{proof}

	\section{Simulation Results}
	\label{sec:sim}

	Numerical simulations are conducted to evaluate the proposed schemes. Four users are considered, each equipped with $N_r = 2$ antennas and supporting two spatial streams. The total bandwidth is
	$20$~MHz with carrier frequency $f^c=28$~GHz and $S=64$, and all users occupy all subcarriers. The transmitted symbols are independently drawn from a normalized 64-QAM constellation, and $B=30$.
	
	Users are placed uniformly at random within a disk of radius $4$~m, located $300$~m from the transmitter.
	A frequency-selective Rician MIMO--OFDM channel with $T$ taps is considered, where the first term corresponds to the dominant line-of-sight (LOS) component and each delayed tap corresponds to a single non-line-of-sight (NLOS) scattering cluster. For user $k$ on subcarrier $s$, the channel matrix $\mathbf{H}_k^s$ is generated as $\mathbf{H}_k^s=
	\sqrt{\frac{\kappa}{\kappa+1}}\sqrt{g_k}\,\mathbf{a}_r(\phi_k^{\text{AoA}})\mathbf{a}_t^H(\theta_k^{\text{AoD}})
	+\sum_{l=1}^{T-1}\sqrt{\frac{1}{\kappa+1}}\sqrt{g_{k,l}}\,\mathbf{a}_r(\phi_{k,l}^{\text{AoA}})\mathbf{a}_t^H(\theta_{k,l}^{\text{AoD}})\,h_{k,l}e^{-\jmath 2\pi l s/S}$ \cite{bjornson2024introduction}. Here, $h_{k,l}\sim\mathcal{CN}(0,1)$. We set $\kappa=10$ and $\chi=3$.
	
	We use the large-scale fading model in \cite{3GPP_TR_36_814_2017}. For user distance $d_k$ (m) and carrier frequency $f^c$ (GHz), the LOS path loss is
	$PL_{\text{dB}}^{\text{LOS}}=22\log_{10}(d_k)+28+20\log_{10}(f^c)+\xi_{\text{LOS}}$, where $\xi_{\text{LOS}}\sim\mathcal{N}(0,\sigma_{\text{LOS}}^2)$ with $\sigma_{\text{LOS}}=5.8$~dB, giving $g_k=10^{-PL_{\text{dB}}^{\text{LOS}}/10}$. For NLOS paths,
	$PL_{\text{dB}}^{\text{NLOS}}=22\log_{10}(d_k)+28+20\log_{10}(f^c)+\xi_{\text{NLOS}}$, where $\xi_{\text{NLOS}}\sim\mathcal{N}(0,\sigma_{\text{NLOS}}^2)$ with $\sigma_{\text{NLOS}}=8.7$~dB, and $g_{k,l}=10^{-PL_{\text{dB}}^{\text{NLOS}}/10}$. The NLOS angles are generated by Gaussian perturbations around the LOS directions:
	$\theta_{k,l}^{\text{AoD}}=\theta_k^{\text{AoD}}+\Delta\theta_{k,l}$ and $\phi_{k,l}^{\text{AoA}}=\phi_k^{\text{AoA}}+\Delta\phi_{k,l}$, where
	$\Delta\theta_{k,l},\Delta\phi_{k,l}\sim\mathcal{N}(0,\sigma_\theta^2)$. The subcarrier-$s$ noise at user $k$ is
	$\mathbf{n}_k^s\sim\mathcal{CN}\!\big(\mathbf{0},\sigma_{\text{noise},k}^{s\,2}\mathbf{I}_{N_r}\big)$, and the noise power spectral density is $-174$~dBm/Hz.
	
	We use a mask that is inactive for $|f|<10.01$~MHz, tightens linearly from $-70$ to $-80$~dBm/$100$~kHz over $|f|\in[10.01,12.5]$~MHz, and remains flat at $-80$~dBm/$100$~kHz for $|f|\ge 12.5$~MHz. We enforce the mask (and the notches over $[-18,-10.01]$~MHz and $[10.01,18]$~MHz) using $90$ uniformly spaced frequency samples per side to form $\mathbf A_n$. For $64$ subcarriers and a per-subcarrier power budget of $30$~dBm, Fig.~\ref{fig:conv_bcd} depicts the per-subcarrier convergence of the proposed BCD algorithm for different numbers of antennas. The sum-MSE decreases monotonically with the iteration index. For the $64$-subcarrier configuration, the per-subcarrier sum-MSE is shown in Fig.~\ref{fig:rate_power_32} as the per-subcarrier transmit power budget is swept from $15$~dBm to $35$~dBm in $5$~dB steps; the average per-subcarrier sum-MSE decreases with increasing transmit power and is further reduced when more antennas are employed.
	The sum-MSE decreases markedly at low power, but the rate of improvement diminishes at higher power as the spectral-mask constraints become active and limit further reduction.
	 When $N_t=16$ and $P^s=30$ dBm, the PSD of the transmitted OFDM symbol from antenna~1, together with the corresponding spectral mask, is shown in Fig.~\ref{fig:rate_power_64}. By enforcing the PSD constraints, the resulting spectrum stays below the mask across the entire OOB region. For comparison, typical ZF and MRT methods are considered, as well as ZF and MRT with frequency notching using identical $\mathbf{A}_n$ as in \cite{taheri2020joint}. It is observed that the proposed method results in the lowest sidelobes. The proposed method not only yields lower OOB emission, but also achieves significantly lower sum-MSE values as depicted in Fig. \ref{fig:mse_power_64}. This is due to the fact that the proposed method deploys an optimized combiner at each user, by which the signal is strengthened and decoded, whereas the benchmarks use fixed (non-optimized) combiners.

	\ifCLASSOPTIONcaptionsoff
	\newpage
	\fi

	\bibliographystyle{IEEEtran}
	
	\bibliography{ref_SDRA}
	\newpage	
\end{document}